\def\version{arxiv}

\documentclass[a4paper,USenglish]{lipics-v2018}

\usepackage{microtype}

\input{preamble}

	\hideLIPIcs
	\nolinenumbers

\title{%
	Nearly-Optimal Mergesorts: \protect\\ 
	Fast, Practical Sorting Methods That \protect\\ Optimally Adapt to Existing Runs%
}

\titlerunning{Nearly-Optimal Mergesort}

\author{J.\ Ian Munro}{University of Waterloo, Canada}{imunro@uwaterloo.ca}{https://orcid.org/0000-0002-7165-7988}{}

\author{Sebastian Wild}{University of Waterloo, Canada}{wild@uwaterloo.ca}{https://orcid.org/0000-0002-6061-9177}{}

\authorrunning{J.\,I.\ Munro and S. Wild}

\Copyright{J.\ Ian Munro and Sebastian Wild}

\subjclass{\ccsdesc[500]{Theory of computation~Sorting and searching}}

\keywords{adaptive sorting, nearly-optimal binary search trees, Timsort}

\funding{%
	This work was supported by the 
	Natural Sciences and Engineering Research Council of Canada 
	and the Canada Research Chairs Programme.
}

\begin{document}

\maketitle

\begin{abstract}
We present two stable mergesort variants, ``peeksort'' and ``powersort'', 
that exploit existing runs and find nearly-optimal merging orders with 
negligible overhead.
Previous methods either require substantial effort for determining the merging order
(Takaoka 2009\ifarxiv{~\cite{Takaoka2009}}{}; Barbay \& Navarro 2013\ifarxiv{~\cite{BarbayNavarro2013}}{}) 
or do not have an optimal worst-case guarantee 
(Peters 2002\ifarxiv{~\cite{Peters2002}}{}; Auger, Nicaud \& Pivoteau 2015\ifarxiv{~\cite{AugerNicaudPivoteau2015}}{}; 
Buss \& Knop 2018\ifarxiv{~\cite{BussKnop2018}}{})\@.
We demonstrate that our methods are competitive in terms of running time with state-of-the-art implementations
of stable sorting methods.
\end{abstract}

\section{Introduction}

Sorting is a fundamental building block for numerous tasks and 
ubiquitous in both the theory and practice of computing.
While practical and theoretically (close-to) optimal 
comparison-based sorting methods are known,
\emph{instance-optimal sorting,} \ie, methods that \emph{adapt} to the actual input
and exploit specific structural properties if present,
is still an area of active research. We survey some recent developments in \wref{sec:related}.

Many different structural properties have been investigated in theory.
Two of them have also found wide adoption in practice, \eg, in Oracle's Java runtime library:
adapting to the presence of duplicate keys 
and using existing sorted segments, called \emph{runs}.
The former is achieved by a so-called fat-pivot partitioning variant of 
quicksort~\cite{BentleyMcIlroy1993},
which is also used in the GNU implementation of \texttt{std::sort} from the 
{C\raisebox{.25ex}{\textsmaller[2]{++}}} STL\@.
It is an \emph{unstable} sorting method, though, \ie, the relative order of 
elements with equal keys might be destroyed in the process.
It is hence used in Java solely for primitive-type arrays.

Making use of existing runs in the input is a well-known option in mergesort; 
\eg, Knuth~\cite{Knuth1998} discusses a bottom-up mergesort variant that does this.
He calls it ``natural mergesort'' and we will use this as an umbrella term for
any mergesort variant that picks up existing
runs in the input (instead of starting blindly with runs of size~$1$).
The Java library uses Timsort~\cite{Peters2002,java2009timsort} 
which is such a natural mergesort originally developed as Python's new library sort.

While fat-pivot quicksort provably adapts
to the \emph{entropy of the multiplicities} of keys~\cite{Wild2018}~--
it is optimal up to a factor of $1.088$ on average with pseudo\-median-of-9 (``ninther'') pivots%
\footnote{%
	The median of three elements is chosen as the pivot, 
	each of which is a median of three other elements.
	This is a good approximation of the median of 9 elements
	and often used as pivot selection rule in library implementations.%
}%
~-- 
Timsort is much more heuristic in nature.
It picks up existing runs and tries to perform merges in a favorable order
(\ie, avoiding merges of runs with very different lengths),
but many desirable guarantees are missing:
Although it was announced as an $O(n\log n)$ worst-case method
with its introduction in Python in 2002~\cite{Peters2002mailinglist},
a rigorous proof of this bound was only given in 2015 by Auger, Nicaud, and Pivoteau~\cite{AugerNicaudPivoteau2015}
and required a rather sophisticated amortization argument.%
\footnote{%
\label{fn:wrong-timsort}%
	A further manifestation of the complexity of Timsort 
	was reported by de~Gouw et al.~\cite{DeGouwBoerBubelHaehnleRotSteinhoefel2017}:
	The original rules to maintain the desired invariant 
	for run lengths on the stack was not sufficient in some cases.
	This (algorithmic!) bug had remained unnoticed until their attempt to formally
	verify the correctness of the Java implementation of Timsort failed because of it.
}
The core complication is that~--
unlike for standard mergesort variants~-- a given element might participate
in more than a logarithmic number of merges.
Indeed, Buss and Knop~\cite{BussKnop2018} have very recently shown 
that for some family of inputs, the average number of merges a single element participates in 
is at least $\bigl(\frac32-o(1)\bigr)\cdot \lg n$.
So in the worst case, Timsort does, \eg, 
(at least) \emph{1.5 times as many element moves as standard mergesort.}

In terms of adapting to existing order, the only proven guarantee for Timsort's running time 
is a trivial $O(n r)$ bound when the input consists of $r$ runs. 
Proving an informative upper bound like $O(n+n \log r )$ has remained elusive,
(although it is conjectured to hold in~\cite{AugerNicaudPivoteau2015} and~\cite{BussKnop2018}).
This is in sharp contrast to available alternatives:
Takaoka~\cite{Takaoka1998,Takaoka2009} and Barbay and Navarro~\cite{BarbayNavarro2013}
independently discovered a sorting method that adapts to the 
\emph{entropy of the distribution of run lengths:}
they sort an input consisting of $r$ runs with respective lengths $L_1,\ldots,L_r$
in time $O\bigl( (\mathcal H(\frac{L_1}n,\ldots,\frac{L_r}n) +1) n \bigr) \subseteq O(n + n\lg r)$,
where $\mathcal H(p_1,\ldots,p_r) = \sum_{i=1}^r p_i \lg(1/p_i)$ is the binary Shannon entropy.
Since $\mathcal H(\frac{L_1}n,\ldots,\frac{L_r}n) n - \Oh(n)$ comparisons are 
necessary for distinct keys, this is optimal up to linear terms.
Their algorithms are also conceptually simple:
find runs in a linear scan, determine an optimal merging order
using a Huffman tree of the run lengths, and execute those merges bottom-up in the tree.
We will refer to this algorithm to determine an optimal merging order 
as \emph{Huffman-Merge}.

Straight-forward implementations of Huffman-Merge
add significant overhead in terms of time and 
	space;
	(finding the Huffman tree requires storing and sorting the run lengths).
	This 
renders these methods uncompetitive to
(reasonable implementations of) elementary sorting methods.
Moreover, Huffman-Merge leads to an \emph{unstable} sorting method
since it merges non-adjacent runs.
The main motivation for the invention of Timsort was 
to find a fast general-purpose sorting method that is \emph{stable}~\cite{Peters2002mailinglist},
and the Java library even dictates the sorting method used for objects to be stable.
We remark that while stability is a much desired feature, 
practical, stable sorting methods do not try to \emph{exploit} the presence of 
duplicate elements to speed up sorting,
and we will focus on the performance for distinct keys in this article.

It is conceptually easy to modify the idea of Takaoka resp.\ Barbay-Navarro 
to sort stably: 
replace the Huffman tree by
an \emph{optimal binary search tree} 
and otherwise proceed 
	as before (using a stable merging procedure).
Since we only have weights at the leaves of the tree, 
we can compute this tree in $O(n + r\log r)$ time using the Hu-Tucker- or Garsia-Wachs-algorithm,
but $r$ can be $\Theta(n)$ and the algorithms are fairly sophisticated,
so this idea seems not very appealing for practical use.

In this paper,
we present two new natural mergesort variants that have the same optimal asymptotic 
running time $O\bigl( (\mathcal H(\frac{L_1}n,\ldots,\frac{L_r}n) +1) n \bigr)$ as Huffman-merge,
but incur much less overhead.
For that, we build upon classic algorithms for computing \emph{nearly-optimal binary search trees}~
\cite{Mehlhorn1984};
but the vital twist for practical methods is to neither explicitly store
the full tree, nor the lengths of all runs at any point in time.
In particular~-- much like Timsort~-- we only store a \emph{logarithmic} number of runs
at any point in time 
(in fact reducing their number from roughly $\log_\varphi \approx1.44 \lg n$ in Timsort to $\lg n$), 
but~-- much \emph{un}like Timsort~-- we retain the guarantee of an optimal merging order up to linear terms.
Our methods require at most $n \lg n + O(n)$ comparison in the worst case
and $\mathcal H(\frac{L_1}n,\ldots,\frac{L_r}n)n +3n$ for an input with runs of lengths $L_1,\ldots, L_r$.

We demonstrate in a running-time study that our methods achieve guaranteed 
(leading-term) optimal adaptive sorting in practice with negligible overhead 
to compute the merge order: 
unlike Timsort, our methods are \emph{not} slower than standard mergesort 
when no existing runs can be exploited.
If existing runs are present, mergesort and quicksort are outperformed by far.
Finally, we show that Timsort is slower than standard mergesort and our new methods
on certain inputs that do have existing runs, but whose lengths pattern hits a weak point
of Timsort's heuristic merging-order rule.

\smallskip\noindent
\textbf{\textsf{Outline:}} The rest of this paper is organized as follows.
In the remainder of this section we survey related work.
\wref{sec:preliminaries} contains notation and 
known results on optimal binary search trees that our work builds on.
The new algorithms and their analytical guarantees are presented in \wref{sec:algorithms}.
\wref{sec:experiments} reports on our running-time study, comparing
the the new methods to state-of-the-art sorting methods.
Finally, \wref{sec:conclusion} summarizes our findings.

\subsection{Adaptive Sorting}
\label{sec:related}

The idea to exploit existing ``structure'' in the input to speed up sorting
dates (at least) back to methods from the 1970s~\cite{Mehlhorn1979} that 
sort faster when the number of inversions is small.
A systematic treatment of this and many further measures of presortedness
(\eg, the number of inversions, the number of runs, and the number of shuffled up-sequences), 
their relation and how to sort \emph{adaptively} \wrt these measures 
are discussed by Estivill-Castro and Wood~\cite{EstivillCastroWood1992}.
While the focus of earlier works is mostly on combinatorial properties of permutations, 
a more recent trend is to consider more fine-grained statistical quantities.
For example, the above mentioned Huffman-Merge adapts to the
\emph{entropy} of the vector of run lengths~\cite{Takaoka1998,Takaoka2009,BarbayNavarro2013}.
Other similar measures are the entropy of the lengths of shuffled up-sequences~\cite{BarbayNavarro2013}
and the entropy of lengths of an LRM-partition~\cite{BarbayFischerNavarro2012}, 
a novel measure that lies between runs and shuffled up-sequences.

For multiset sorting, the fine-grained measure, the \emph{entropy of the multiplicities}, 
has been considered instead of the \emph{number} of unique values already in early work 
in the field (\eg~\cite{MunroSpira1976,Sedgewick1977equals}).
A more recent endeavor has been to find sorting methods that optimally adapt to
\emph{both presortedness and repeated values}.
Barbay, Ochoa, and Satti refer to this as \emph{synergistic sorting}~\cite{BarbayOchoaSatti2017}
and present an algorithm based on quicksort that is optimal up to a constant factor.
The method's practical performance is unclear.

We remark that (unstable) multiset sorting is the \emph{only} problem from the above list
for which a theoretically optimal algorithm has found wide-spread adoption in
programming libraries:
quicksort is known to almost optimally adapt to the entropy of multiplicities 
on average~\cite{Wegner1985,SedgewickBentley2002,Wild2018},
when elements equal to the pivot are excluded from recursive calls (fat-pivot partitioning).
Supposedly, sorting is so fast to start with that further improvements from
exploiting specific input characteristics are only fruitful if they can be realized
with minimal additional overhead.
Indeed, for algorithms that adapt to the number of inversions, 
Elmasry and Hammad~\cite{ElmasryHammad2009} found that 
the adaptive methods could only compete with good implementations of elementary sorting algorithms 
in terms of running time
for inputs with extremely few inversions (fewer than 1.5\%).
Translating the theoretical guarantees of adaptive sorting
into practical, efficient methods is an ongoing challenge.

\subsection{Lower bound}
\label{sec:lower-bound}

How much does it help for sorting an array $A[1..n]$ to
know that it contains $r$ runs of respective sizes $L_1,\ldots,L_r$,
\ie, to know the relative order of $A[1..L_1]$, $A[L_1+1..L_1+L_2]$ etc.?
If we assume distinct elements, a simple counting argument shows that there
are $\binom n{L_1,\ldots,L_r}$ permutations that are compatible with this setup.
(the number of ways to partition $n$ keys into $r$ subsets of given sizes.)
We thus need $\lg(n!) - \sum_{i=1}^r \lg (L_i!) = \mathcal H(\frac{L_1}n,\ldots,\frac{L_r}n) n - O(n)$ 
comparisons to sort such an input.
A formal argument for this lower bound is given by Barbay and Navarro~\cite{BarbayNavarro2013} 
in the proof of their Theorem~2. 

\subsection{Results on Timsort and stack-based mergesort}

Its good performance in running-time studies especially on partially sorted inputs
have lead to the adoption of Timsort in several programming libraries,
but as mentioned in the introduction, 
the complexity of the algorithm has precluded researchers from proving 
interesting adaptivity guarantees.
To make progress towards these, simplified variations of Timsort have been 
considered~\cite{AugerNicaudPivoteau2015,BussKnop2018}.
All of those methods work by maintaining a stack of runs yet to be merged
and proceed as follows: 
They find the next run in the input and push it onto the stack.
Then they consider the top $k$ elements on the stack (for $k$ a small constant like 3 or 4)
and decide based on these if any pair of them is to be merged.
If so, the two runs in the stack are replaced with the merged result and the rule is applied
repeatedly until the stack satisfies some invariant.
The invariant is chosen so as to keep the height of the stack small (logarithmic in $n$).

The simplest version, ``$\alpha$-stack sort''~\cite{AugerNicaudPivoteau2015}, 
merges the topmost two runs until the run lengths in the stack grow at least by a factor of $\alpha$,
(\eg, $\alpha=2$).
This method can lead to imbalanced merges (and hence runtime $\omega(n\log r)$~\cite{BussKnop2018};
the authors of~\cite{AugerNicaudPivoteau2015} also point this out in their conclusion):
if the next run is much larger than what is on the stack, a much more balanced
merging order results from first merging stack elements until they are at least as big
as the new run.
This modification is called ``$\alpha$-merge sort'', which achieves a worst-case 
guarantee of $O(n + n \log r)$, but the constant is provably not optimal~\cite{BussKnop2018}
(for any $\alpha>1$).

Timsort is quite similar to $\alpha$-merge sort for $\alpha=\varphi$ (the golden ratio)
by forcing the run lengths to grow at least like Fibonacci numbers.
The details of the rule are given in~\cite{AugerNicaudPivoteau2015} or~\cite{BussKnop2018} and
are quite intricate~-- 
and were indeed wrong in the first (widely-used) version of Timsort 
(see \wtpref{fn:wrong-timsort}).
While it is open if Timsort always runs in $O(n+n \log r)$ time,
Buss and Knop gave a family of inputs for which Timsort does asymptotically at least $1.5$ times
the required effort (in terms of merge costs, see \wref{sec:merge-costs}), and hence
proved that Timsort~-- like $\alpha$-merge sort~-- is \emph{not} optimally adaptive
even to the \emph{number} of runs $r$, not to speak of the entropy of the run lengths.

\section{Preliminaries}
\label{sec:preliminaries}

We implicitly assume that we are sorting an array $A[1..n]$ of $n$ elements.
By $\mathcal H$, we denote the binary Shannon entropy, \ie,
for $p_1,\ldots,p_m\in[0,1]$ with $p_1+\cdots+p_m = 1$ we let
$\mathcal H(p_1,\ldots,p_m) = \sum p_i \lg(1/p_i)$, where
$\lg=\log_2$.

We will always let $r$ denote the \emph{number of runs} in the input and 
$L_1,\ldots,L_r$ their respective lengths with $L_1+\cdots+L_r = n$.
In the literature, a \emph{run} usually means a maximal (contiguous) 
	weakly increasing%
	\footnote{%
		We use ``weakly increasing'' to mean ``nondecreasing''.\\
		\textsmaller[2]{(Isn't it better to say what we mean instead of not saying what we don't mean?)}
	}
	region,
but we adopt the convention from Timsort in this paper: a run is either
a maximal \emph{weakly increasing} region \emph{or} a \emph{strictly decreasing} region.
Decreasing runs are immediately reversed;
allowing only strict decreasing runs makes their \emph{stable} reversal trivial.
The algorithms are not directly affected by different conventions for what a ``run'' is;
they only rely on a unique partition of the input into sorted segments that can be found
by sequential scans.

\subsection{Nearly-Optimal Binary Search Trees}
\label{sec:optimal-BSTs}

In the \emph{optimal binary search tree problem,} 
we are given probabilities $\beta_1,\ldots,\beta_m$ to access the $m$
keys $K_1<\cdots<K_m$ (internal nodes) and probabilities
$\alpha_0,\ldots,\alpha_m$ to access the gaps (leaves) between these keys 
(setting $K_0 = -\infty$ and $K_{m+1}=+\infty$)
and we are interested in the binary search tree that minimizes the expected
search cost $C$, \ie, 
the expected number of (ternary) comparisons when access follow the given distribution.%
\footnote{%
	We deviate from the literature convention and use $m$ to denote the number of 
	keys to avoid confusion with $n$, the length of the arrays to sort,
	in the rest of the paper.
}
Nagaraj~\cite{Nagaraj1997} surveys various versions of the problem.
We confine ourselves to \emph{approximation algorithms} here.
Moreover, we only need the special case of \emph{alphabetic trees}
where all $\beta_j = 0$.

The following methods apply to the general problem, but we 
present them for the case of \emph{nearly-optimal alphabetic trees.}
So in the following let $\alpha_0,\ldots,\alpha_m$ with $\sum_{i=0}^m \alpha_i = 1$ be given.
If the details are done right, a greedy top-down approach produces provably good 
search trees~\cite{Bayer1975,Mehlhorn1977}:
choose the boundary closest to $\frac12$ as the bisection at the root 
(\emph{``weight-balancing heuristic''}).
Mehlhorn~\cite[\S III.4.2]{Mehlhorn1984} discusses two algorithms for 
nearly-optimal binary search trees that follow this scheme:
``Method 1'' is the straight-forward recursive application of the above rule,
whereas ``Method 2'' (\emph{``bisection heuristic''}) 
continues by strictly halving the \emph{original} interval in the recursive calls;
see \wref{fig:nearly-optimal-trees}.

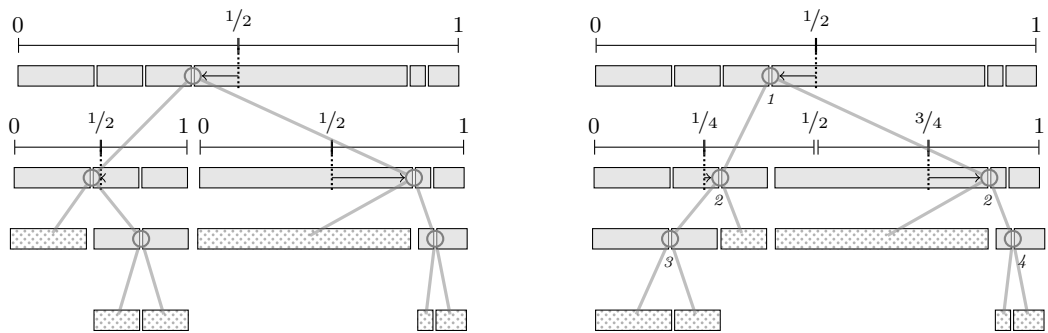
\begin{figure}[tbph]
	\begin{center}
	\begin{tikzpicture}[
			yscale=.27,
			xscale=.2,
			every node/.style={inner sep=2pt,font=\footnotesize},
			tree/.style={
				circle,draw,black!50,thick,minimum size=6pt,inner sep=0pt,fill=black!10,fill opacity=.5
			},
			orgrun/.style={pattern=crosshatch dots,pattern color=black!30},
	]
		\def\delta{0.2}
		\def\intskip{2}
		\useasboundingbox (-.5,4) rectangle (67.5,-11);
		
		\draw[fill=black!10] (0,0) coordinate (l)
			\foreach \d/\i in {5,3,3,14,1,2} {
				rectangle 
				++(\d,1) ++(\delta,-1) 
			}
			++(-\delta,0) coordinate (r)
		;
		\coordinate (m) at ($(l)!.5!(r)$);
		\draw[|-|] ($(l)+(0,\intskip)$) node[above=3pt] {0} -- ($(m)+(0,\intskip)$) node[above=3pt] {\nicefrac12} ;
		\draw[|-|] ($(r)+(0,\intskip)$) node[above=3pt] {1} -- ($(m)+(0,\intskip)$) ;
		\node[tree] (n0) at (5+3+3+2.5*\delta,.5) {};
		\draw[thick,densely dotted] (m) ++ (0,-.1) -- ++(0,2) ;
		\draw[->] ($(l)!.5!(r)$) ++(0,.5) -- (n0) ;
		
		\begin{scope}[shift={(-.25,-5)}]
			\draw[fill=black!10] (0,0) coordinate (l)
				\foreach \d/\i in {5,3,3} {
					rectangle 
					++(\d,1) ++(\delta,-1) 
				}
				++(-\delta,0) coordinate (r)
			;
			\coordinate (m) at ($(l)!.5!(r)$);
			\draw[|-|] ($(l)+(0,\intskip)$) node[above=3pt] {0} -- ($(m)+(0,\intskip)$) node[above=3pt] {\nicefrac12} ;
			\draw[|-|] ($(r)+(0,\intskip)$) node[above=3pt] {1~~} -- ($(m)+(0,\intskip)$) ;
			\node[tree] (n00) at (5+0.5*\delta,.5) {};
			\draw[thick,densely dotted] (m) ++ (0,-.1) -- ++(0,2) ;
			\draw[->] ($(l)!.5!(r)$) ++(0,.5) -- (n00) ;
			\draw[black!50,very thick,opacity=.5] (n0) to (n00);
		\end{scope}
		
		\begin{scope}[shift={(-.5,-8)}]
			\draw[fill=black!10,orgrun] (0,0) coordinate (l)
				\foreach \d/\i in {5} {
					rectangle node (n000) {}
					++(\d,1) ++(\delta,-1) 
				}
				++(-\delta,0) coordinate (r)
			;
			\draw[black!50,very thick,opacity=.5] (n00) to (n000);
		\end{scope}
		
		\begin{scope}[shift={(-.5+.5+5,-8)}]
			\draw[fill=black!10] (0,0) coordinate (l)
				\foreach \d/\i in {3,3} {
					rectangle node (n000) {}
					++(\d,1) ++(\delta,-1) 
				}
				++(-\delta,0) coordinate (r)
			;
			\node[tree] (n001) at (3+0.5*\delta,.5) {};
			\draw[black!50,very thick,opacity=.5] (n00) to (n001);
		\end{scope}
		\begin{scope}[shift={(+5,-12)}]
			\draw[fill=black!10,orgrun] (0,0) coordinate (l)
				\foreach \d/\i in {3} {
					rectangle node (n0010) {}
					++(\d,1) ++(\delta,-1) 
				}
				++(-\delta,0) coordinate (r)
			;
			\draw[black!50,very thick,opacity=.5] (n001) to (n0010);
		\end{scope}
		\begin{scope}[shift={(+5+3+\delta,-12)}]
			\draw[fill=black!10,orgrun] (0,0) coordinate (l)
				\foreach \d/\i in {3} {
					rectangle node (n0011) {}
					++(\d,1) ++(\delta,-1) 
				}
				++(-\delta,0) coordinate (r)
			;
			\draw[black!50,very thick,opacity=.5] (n001) to (n0011);
		\end{scope}
	
		\begin{scope}[shift={(0.25+5+3+3+3.5*\delta,-5)}]
			\draw[fill=black!10] (0,0) coordinate (l)
				\foreach \d/\i in {14,1,2} {
					rectangle 
					++(\d,1) ++(\delta,-1) 
				}
				++(-\delta,0) coordinate (r)
			;
			\coordinate (m) at ($(l)!.5!(r)$);
			\draw[|-|] ($(l)+(0,\intskip)$) node[above=3pt] {~0} -- ($(m)+(0,\intskip)$) node[above=3pt] {\nicefrac12} ;
			\draw[|-|] ($(r)+(0,\intskip)$) node[above=3pt] {1} -- ($(m)+(0,\intskip)$) ;
			\node[tree] (n01) at (14+0.5*\delta,.5) {};
			\draw[thick,densely dotted] (m) ++ (0,-.1) -- ++(0,2) ;
			\draw[->] ($(l)!.5!(r)$) ++(0,.5) -- (n01) ;
			\draw[black!50,very thick,opacity=.5] (n0) to (n01);
		\end{scope}	
		
		\begin{scope}[shift={(.125+5+3+3+3.5*\delta,-8)}]
			\draw[fill=black!10,orgrun] (0,0) coordinate (l)
				\foreach \d/\i in {14} {
					rectangle node (n010) {}
					++(\d,1) ++(\delta,-1) 
				}
				++(-\delta,0) coordinate (r)
			;
			\draw[black!50,very thick,opacity=.5] (n01) to (n010);
		\end{scope}
		
		\begin{scope}[shift={(.25+5+3+3+14+5.5*\delta,-8)}]
			\draw[fill=black!10] (0,0) coordinate (l)
				\foreach \d/\i in {1,2} {
					rectangle node (n000) {}
					++(\d,1) ++(\delta,-1) 
				}
				++(-\delta,0) coordinate (r)
			;
			\node[tree] (n011) at (1+0.5*\delta,.5) {};
			\draw[black!50,very thick,opacity=.5] (n01) to (n011);
		\end{scope}
		\begin{scope}[shift={(5+3+3+14+6.5*\delta,-12)}]
			\draw[fill=black!10,orgrun] (0,0) coordinate (l)
				\foreach \d/\i in {1} {
					rectangle node (n0110) {}
					++(\d,1) ++(\delta,-1) 
				}
				++(-\delta,0) coordinate (r)
			;
			\draw[black!50,very thick,opacity=.5] (n011) to (n0110);
		\end{scope}
		\begin{scope}[shift={(5+3+3+14+1+7.5*\delta,-12)}]
			\draw[fill=black!10,orgrun] (0,0) coordinate (l)
				\foreach \d/\i in {2} {
					rectangle node (n0111) {}
					++(\d,1) ++(\delta,-1) 
				}
				++(-\delta,0) coordinate (r)
			;
			\draw[black!50,very thick,opacity=.5] (n011) to (n0111);
		\end{scope}

	\begin{scope}[shift={(38,0)}]
	
		\draw[fill=black!10] (0,0) coordinate (l)
			\foreach \d/\i in {5,3,3,14,1,2} {
				rectangle 
				++(\d,1) ++(\delta,-1) 
			}
			++(-\delta,0) coordinate (r)
		;
		\coordinate (m) at ($(l)!.5!(r)$);
		\coordinate (MM) at ($(l)!.5!(r)$);
		\draw[|-|] ($(l)+(0,\intskip)$) node[above=3pt] {0} -- ($(m)+(0,\intskip)$) node[above=3pt] {\nicefrac12} ;
		\draw[|-|] ($(r)+(0,\intskip)$) node[above=3pt] {1} -- ($(m)+(0,\intskip)$) ;
		\node[tree] (n0) at (5+3+3+2.5*\delta,.5) {}; 
		\draw[thick,densely dotted] (m) ++ (0,-.1) -- ++(0,2) ;
		\draw[->] ($(l)!.5!(r)$) ++(0,.5) -- (n0) ;
		
		\begin{scope}[shift={(-.1,-5)}]
			\draw[fill=black!10] (0,0) coordinate (l)
				\foreach \d/\i in {5,3,3} {
					rectangle 
					++(\d,1) ++(\delta,-1) 
				}
				++(-\delta,0) coordinate (rr)
			;
			\coordinate (r) at ($(MM)+(-.1,-5)$) ;
			\coordinate (m) at ($(l)!.5!(r)$);
			\draw[|-|] ($(l)+(0,\intskip)$) node[above=3pt] {0} -- ($(m)+(0,\intskip)$) node[above=3pt] {\nicefrac14} ;
			\draw[|-|] ($(r)+(0,\intskip)$) node[above=3pt] {\nicefrac12} -- ($(m)+(0,\intskip)$) ;
			\node[tree] (n00) at (5+3+1.5*\delta,.5) {};
			\draw[thick,densely dotted] (m) ++ (0,-.1) -- ++(0,2) ;
			\draw[->] ($(l)!.5!(r)$) ++(0,.5) -- (n00) ;
			\draw[black!50,very thick,opacity=.5] (n0) to (n00);
		\end{scope}
		
		\begin{scope}[shift={(-.2,-8)}]
			\draw[fill=black!10] (0,0) coordinate (l)
				\foreach \d/\i in {5,3} {
					rectangle 
					++(\d,1) ++(\delta,-1) 
				}
				++(-\delta,0) coordinate (r)
			;
			\node[tree] (n000) at (5+0.5*\delta,.5) {};
			\draw[black!50,very thick,opacity=.5] (n00) to (n000);
		\end{scope}
		
		\begin{scope}[shift={(-.25+.5+5+3,-8)}]
			\draw[fill=black!10,orgrun] (0,0) coordinate (l)
				\foreach \d/\i in {3} {
					rectangle node (n001) {}
					++(\d,1) ++(\delta,-1) 
				}
				++(-\delta,0) coordinate (r)
			;
			\draw[black!50,very thick,opacity=.5] (n00) to (n001);
		\end{scope}
		\begin{scope}[shift={(0,-12)}]
			\draw[fill=black!10,orgrun] (0,0) coordinate (l)
				\foreach \d/\i in {5} {
					rectangle node (n0010) {}
					++(\d,1) ++(\delta,-1) 
				}
				++(-\delta,0) coordinate (r)
			;
			\draw[black!50,very thick,opacity=.5] (n000) to (n0010);
		\end{scope}
		\begin{scope}[shift={(+5+\delta,-12)}]
			\draw[fill=black!10,orgrun] (0,0) coordinate (l)
				\foreach \d/\i in {3} {
					rectangle node (n0011) {}
					++(\d,1) ++(\delta,-1) 
				}
				++(-\delta,0) coordinate (r)
			;
			\draw[black!50,very thick,opacity=.5] (n000) to (n0011);
		\end{scope}
	
		\begin{scope}[shift={(0.1+5+4+2+3.5*\delta,-5)}]
			\draw[fill=black!10] (0,0) coordinate (ll)
				\foreach \d/\i in {14,1,2} {
					rectangle 
					++(\d,1) ++(\delta,-1) 
				}
				++(-\delta,0) coordinate (r)
			;
			\coordinate (l) at ($(MM)+(.1,-5)$);
			\coordinate (m) at ($(l)!.5!(r)$);
			\draw[|-|] ($(l)+(0,\intskip)$) node[above=3pt] {} -- ($(m)+(0,\intskip)$) node[above=3pt] {\nicefrac34} ;
			\draw[|-|] ($(r)+(0,\intskip)$) node[above=3pt] {1} -- ($(m)+(0,\intskip)$) ;
			\node[tree] (n01) at (14+0.5*\delta,.5) {};
			\draw[thick,densely dotted] (m) ++ (0,-.1) -- ++(0,2) ;
			\draw[->] ($(l)!.5!(r)$) ++(0,.5) -- (n01) ;
			\draw[black!50,very thick,opacity=.5] (n0) to (n01);
		\end{scope}	
		
		\begin{scope}[shift={(.125+5+4+2+3.5*\delta,-8)}]
			\draw[fill=black!10,orgrun] (0,0) coordinate (l)
				\foreach \d/\i in {14} {
					rectangle node (n010) {}
					++(\d,1) ++(\delta,-1) 
				}
				++(-\delta,0) coordinate (r)
			;
			\draw[black!50,very thick,opacity=.5] (n01) to (n010);
		\end{scope}
		
		\begin{scope}[shift={(.25+5+4+2+14+5.5*\delta,-8)}]
			\draw[fill=black!10] (0,0) coordinate (l)
				\foreach \d/\i in {1,2} {
					rectangle node (n0111) {}
					++(\d,1) ++(\delta,-1) 
				}
				++(-\delta,0) coordinate (r)
			;
			\node[tree] (n011) at (1+0.5*\delta,.5) {};
			\draw[black!50,very thick,opacity=.5] (n01) to (n011);
		\end{scope}
		\begin{scope}[shift={(5+4+2+14+6.5*\delta,-12)}]
			\draw[fill=black!10,orgrun] (0,0) coordinate (l)
				\foreach \d/\i in {1} {
					rectangle node (n0110) {}
					++(\d,1) ++(\delta,-1) 
				}
				++(-\delta,0) coordinate (r)
			;
			\draw[black!50,very thick,opacity=.5] (n011) to (n0110);
		\end{scope}
		\begin{scope}[shift={(5+4+2+14+1+7.5*\delta,-12)}]
			\draw[fill=black!10,orgrun] (0,0) coordinate (l)
				\foreach \d/\i in {2} {
					rectangle node (n0111) {}
					++(\d,1) ++(\delta,-1) 
				}
				++(-\delta,0) coordinate (r)
			;
			\draw[black!50,very thick,opacity=.5] (n011) to (n0111);
		\end{scope}
		
		\node[yshift=-2ex,xshift=-.2ex] at (n000) {\smaller[2]\itshape 3};
		\node[yshift=-2ex,xshift=-.2ex] at (n00)  {\smaller[2]\itshape 2};
		\node[yshift=-2ex] at (n0)   {\smaller[2]\itshape 1};
		\node[yshift=-2ex,xshift=-.2ex] at (n01)  {\smaller[2]\itshape 2};
		\node[yshift=-2ex,xshift=.8ex] at (n011) {\smaller[2]\itshape 4};
		
	\end{scope}
	
	\end{tikzpicture}
	\end{center}
	\caption{%
		The two versions of weight-balancing for computing nearly-optimal alphabetic trees.
		The gap probabilities are proportional to $5,3,3,14,1,2$.
		\textbf{\sffamily Left:} Mehlhorn's ``Method~1'' chooses the split closest to the midpoint of the 
		subtree's actual weights ($\nicefrac12$ after renormalizing).
		\textbf{\sffamily Right:} ``Method~2'' continues to cut the original interval in half, irrespective
		of the total weight of the subtrees.
		The italic numbers are the powers of the nodes (see \wtpref{def:node-power}).
	}
\label{fig:nearly-optimal-trees}
\end{figure}

Method~1 was proposed in~\cite{WalkerGotlieb1972} and analyzed in \cite{Mehlhorn1975,Bayer1975};
Method~2 is discussed in~\cite{Mehlhorn1977}.
While Method~1 is arguably more natural, Method~2 has the advantage to yield
splits that are predictable without going through all steps of the recursion.
Both methods can be implemented to run in time $O(m)$ and yield very good trees.
(Recall that in the case $\beta_j = 0$ the classic information-theoretic 
argument dictates $C\ge \mathcal H$; 
Bayer~\cite{Bayer1975} gives lower bounds in the general case.)

\begin{theorem}[Nearly-Optimal BSTs]
\label{thm:nearly-opt-trees}
	Let $\alpha_0,\beta_1,\alpha_1,\ldots,\beta_m,\alpha_m \in [0,1]$ with $\sum \alpha_i+\sum\beta_j = 1$ 
	be given and let $\mathcal H = \sum_{i=0}^m \alpha_i \lg(1/\alpha_i) + \sum_{j=1}^m \beta_j \lg(1/\beta_j)$.
\begin{enumerate}[(i)]
\item 
	Method~1 yields a tree with search cost 
	$C \le \mathcal H + 2$.
	\cite[Thm\,4.8]{Bayer1975}
\item 
	If all $\beta_j = 0$,
	Method~1 yields a tree with search cost 
	$C \le \mathcal H + 2 - (m+3) \alpha_{\mathrm{min}}$,\\
	where $\alpha_{\mathrm{min}} = \min \{\alpha_0,\ldots,\alpha_m\}$.
	\cite{Horibe1977}
\item 
	Method~2 yields a tree with search cost
	$C \le \mathcal H + 1 + \sum \alpha_i$.
	\cite{Mehlhorn1977}
\end{enumerate}
\end{theorem}

\subsection{Merge Costs}
\label{sec:merge-costs}

In this paper, we are primarily concerned with finding a good \emph{order} of binary merges 
for the existing runs in the input.
Following \cite{AugerNicaudPivoteau2015} and \cite{BussKnop2018},
we will define the \emph{merge cost} $M$ for merging two runs of lengths $m$ resp.\ $n$ 
as $M = m+n$, \ie, the size of the result.
This quantity has been studied earlier by 
Golin and Sedgewick~\cite{GolinSedgewick1993} without giving it a name.

Merge costs abstract away from key comparisons and element moves
and simplify computations (see next subsection).
Since any merge has to move most elements (except for rare lucky cases),
and the average number of comparisons using standard merge routines
is $m+n - \bigl(\frac m{n+1} + \frac{n}{m+1}\bigr)$, 
merge costs are a reasonable approximation,
in particular when $m$ and $n$ are roughly equal.
They always yield an upper bound for both the number of comparisons and moves.

\subsection{Merge Trees}
\label{sec:merge-trees}

Let $L_1,\ldots, L_r$ with $\sum L_i = n$ be the lengths of the runs in the input.
Any natural mergesort can be described as a rule to select some of the remaining runs, 
which are then merged and replaced by the merge result. If we always merge \emph{two} runs
this corresponds to a binary tree with the original runs at leaves $\leafnode1,\ldots,\leafnode{\like1r}$. 
Internal nodes correspond to the result of merging their children.
If we assign to internal node \internalnode j
the size $M_j$ of the (intermediate) merge result it represents,
then the overall merge cost is exactly $M = \sum_{\internalnodescript j} M_j$ (summing
over all internal nodes).
\wref{fig:nearly-optimal-trees} shows two examples of merge trees;
the merge costs are given by adding up all gray areas,%
\footnote{%
	The left tree is obviously better here
	and this is a typical outcome.
	But there are also inputs where Method 2 yields a 
	better tree than Method 1.
}
(ignoring the dotted leaves).

Let $d_i$ be the \emph{depth} of leaf \leafnode i (corresponding to the run of length $L_i$),
where depth is the number of edges on the path to the root.
Every element in the $i$th run is counted exactly $d_i$ times in $\sum_{\internalnodescript j} M_j$,
so we have $M = \sum_{i=1}^r d_i\cdot L_i$.
Dividing by $n$ yields $M/n = \sum_{i=1}^r d_i\cdot \alpha_i$ for $\alpha_i = L_i / n$, 
which happens to be the expected search time $C$ in the merge tree
if \leafnode i is requested with probability $\alpha_i$ for $i=1,\ldots,r$.
So the minimal-cost merge tree for given run lengths $L_1,\ldots,L_r$ 
is the optimal alphabetic tree for leaf probabilities $\frac{L_1}n,\ldots,\frac{L_r}n$
and it holds
$M \ge \mathcal H(\frac{L_1}n,\ldots,\frac{L_m}n) n$.
For distinct keys, the lower bound on comparisons (\wref{sec:lower-bound}) 
coincides up to linear terms 
with this lower bound.

Combining this fact with the linear-time methods for nearly-optimal binary search trees 
from \wref{sec:optimal-BSTs}
immediately gives a stable sorting method that adapts optimally to existing runs
up to an $O(n)$ term.
We call such a method a \emph{nearly-optimal (natural) mergesort}.
	A~direct implementation of this idea needs $\Theta(r)$ space to store $L_1,\ldots,L_r$ 
	and the merge tree and does an extraneous pass over the data 
	just to determine the run lengths.
	The purpose of this paper is to show that we can make the overhead 
	for finding a nearly-optimal merging order
	negligible in time and space.

\section{Nearly-Optimal Merging Orders}
\label{sec:algorithms}

We now describe two sorting methods that simulate nearly-optimal search tree algorithms
to compute nearly-optimal merging orders, but do so without ever storing the full merge tree 
or even the run lengths.

\subsection{Peeksort: A Simple Top-Down Method}

The first method is similar to standard top-down mergesort
in that it implicitly constructs a merge tree on the call stack.
Instead of blindly cutting the input in half, however, we mimic Mehlhorn's Method~1.
For that we need the \emph{run boundary closest to the middle} of the input: 
this will become the root of the merge tree.
Since we want to detect existing runs anyway at some point, 
we start by finding the run that contains the middle position.
The end point closer to the middle determines the top-level cut and 
we recursively sort the parts left and right of it.
A final merge completes the sort.

To avoid redundant scans, we pass on the information about already detected runs.
In the general case, we are sorting a range $A[\ell..r]$ whose 
prefix $A[\ell..e]$ and suffix $A[s..r]$ are runs.
Depending on whether the middle is contained in one of those runs, 
we have one of four different cases; apart from that
the overall procedure (\wref{alg:peeksort}) is quite straight-forward.

\begin{algorithm}[tbph]
\vspace*{-2ex}
\small
\tikzset{every node/.style={font=\scriptsize}}
	\begin{codebox}
		\Procname{$\proc{PeekSort}(A[\ell..r],e,s)$}
		\li \kw{if} $e \isequal r $ or $s \isequal \ell$ \kw{then} \Return
		\li $m \gets \ell + \bigl\lfloor \frac{r-\ell}2 \bigr\rfloor$
		\li \If $m \le e$   
			\>\>\>\>\>\Comment
			\begin{tikzpicture}[scale=.35,baseline=.3ex]
				\draw[fill=black!10] (0,0) rectangle (13,1) ;
				\draw[fill=black!10] (18,0) rectangle (20,1) ;
				\draw (0,0) rectangle (10,1) rectangle (20,0) ;
				\node[anchor=base] at ( 0.5,.25) {$\ell$} ;
				\node[anchor=base] at (19.5,.25) {$r$} ;
				\node[anchor=base] at (12.5,.25) {$e$} ;
				\node[anchor=base] at (18.5,.25) {$s$} ;
				\node[overlay] at (10,1.3) {$m$} ;
			\end{tikzpicture}
		\Then
			\li $\proc{PeekSort}(A[e+1..r],e+1,s)$
			\li $\proc{Merge}(A[\ell..e],A[e+1..r])$
		\li \Else \kw{if} $m \ge s$
			\>\>\>\>\Comment
			\begin{tikzpicture}[scale=.35,baseline=.3ex]
				\draw[fill=black!10] (0,0) rectangle (3,1) ;
				\draw[fill=black!10] (8,0) rectangle (20,1) ;
				\draw (0,0) rectangle (10,1) rectangle (20,0) ;
				\node[anchor=base] at ( 0.5,.25) {$\ell$} ;
				\node[anchor=base] at (19.5,.25) {$r$} ;
				\node[anchor=base] at (2.5,.25) {$e$} ;
				\node[anchor=base] at (8.5,.25) {$s$} ;
				\node[overlay] at (10,1.3) {$m$} ;
			\end{tikzpicture}
			\li $\proc{PeekSort}(A[\ell..s-1], e,s-1)$
			\li $\proc{Merge}(A[\ell..s-1],A[s..r])$
		\li \Else
			\zi \Comment Find existing run $A[i..j]$ containing position $m$
			\li $i \gets \proc{ExtendRunLeft}(A[m], \ell)$; \quad
				$j \gets \proc{ExtendRunRight}(A[m], r)$
			\li \kw{if} $i\isequal \ell$ and $j\isequal r$ \Return
			\li \If $m - i < j - m$
			\>\>\>\>\Comment
			\begin{tikzpicture}[scale=.35,baseline=.3ex]
				\draw[fill=black!10] (0,0) rectangle (3,1) ;
				\draw[fill=black!10] (18,0) rectangle (20,1) ;
				\draw[fill=black!10] (8.5,0) rectangle (16,1) ;
				\draw (0,0) rectangle (10,1) rectangle (20,0) ;
				\node[anchor=base] at ( 0.5,.25) {$\ell$} ;
				\node[anchor=base] at (19.5,.25) {$r$} ;
				\node[anchor=base] at (2.5,.25) {$e$} ;
				\node[anchor=base] at (18.5,.25) {$s$} ;
				\node[anchor=base] at (9.0,.25) {$i$} ;
				\node[anchor=base] at (15.5,.25) {$j$} ;
				\node[overlay] at (10,1.3) {$m$} ;
			\end{tikzpicture}
			\Then
				\li $\proc{PeekSort}(A[\ell..i-1],e,i-1)$
				\li $\proc{PeekSort}(A[i..r],j,s)$
				\li $\proc{Merge}(A[\ell..i-1],A[i..r])$
			\li \Else
			\>\>\>\Comment
			\begin{tikzpicture}[scale=.35,baseline=.3ex]
				\draw[fill=black!10] (0,0) rectangle (3,1) ;
				\draw[fill=black!10] (18,0) rectangle (20,1) ;
				\draw[fill=black!10] (5.5,0) rectangle (12,1) ;
				\draw (0,0) rectangle (10,1) rectangle (20,0) ;
				\node[anchor=base] at ( 0.5,.25) {$\ell$} ;
				\node[anchor=base] at (19.5,.25) {$r$} ;
				\node[anchor=base] at (2.5,.25) {$e$} ;
				\node[anchor=base] at (18.5,.25) {$s$} ;
				\node[anchor=base] at (6,.25) {$i$} ;
				\node[anchor=base] at (11.5,.25) {$j$} ;
				\node[overlay] at (10,1.3) {$m$} ;
			\end{tikzpicture}
				\li $\proc{PeekSort}(A[\ell..j],e,i)$
				\li $\proc{PeekSort}(A[j+1..r],j+1,s)$
				\li $\proc{Merge}(A[\ell..j],A[j+1..r])$
	\end{codebox}
	\vspace*{-1ex}
	\caption{%
		Peeksort: A simple top-down version of nearly-optimal natural mergesort.
		The initial call is $\proc{PeekSort}(A[1..n],1,n)$.
		Procedures \proc{ExtendRunLeft} (-\proc{Right}) scan left (right) starting at $A[m]$
		as long as the run continues (and we did not cross the second parameter).%
	}
	\label{alg:peeksort}
	\vspace*{-3ex}
\end{algorithm}

The following theorem shows that \proc{PeekSort} is indeed a nearly-optimal mergesort.
Unlike previous such methods, its code has very little overhead
(in terms of both time and space) in comparison
with a standard top-down mergesort,
so it is a promising method for a practical nearly-optimal mergesort.

\begin{theorem}
\label{thm:peeksort}
	The merge cost of \proc{PeekSort} on an input consisting of $r$ runs with 
	respective lengths $L_1,\ldots,L_r$ is at most $\mathcal H(\frac{L_1}n,\ldots,\frac{L_m}n)n + 2n - (r+2)$,
	the number of comparisons is at most $\mathcal H(\frac{L_1}n,\ldots,\frac{L_m}n)+3n - (2r+3)$.
	Both is optimal up to $O(n)$ terms (in the worst case).
\end{theorem}

\begin{proof}
	The recursive calls of \wref{alg:peeksort} produce the same tree as Mehlhorn's Method~1
	with input $(\alpha_0,\ldots,\alpha_m) = (\frac{L_1}n,\ldots,\frac{L_r}n)$ (\ie, $m=r-1$)
	and $\beta_j=0$.
	By \wref{thm:nearly-opt-trees}--(ii), the search costs in this tree are 
	$C \le \mathcal H+2 - (m+3)\alpha_{\mathrm{min}}$ with $\mathcal H = \mathcal H(\frac{L_1}n,\ldots,\frac{L_m}n)$.
	Since $L_j \ge 1$, we have $\alpha_{\mathrm{min}} \ge\frac1n$.
	As argued in \wref{sec:merge-trees}, the overall merge costs are then given by
	$M = C n \le \mathcal H n + 2n - (r+2)$, which is within $O(n)$ of the lower bound for $M$.
	
	We save at least one comparison per merge since merging runs of lengths $m$ and $n$ 
	requires at most $n+m-1$ comparisons. In total, we do exactly $r-1$ merge operations. 
	Apart from merging, we need a total of $n-1$ additional comparisons to detect the existing runs
	in the input.
	Barbay and Navarro~\cite[Thm.\,2]{BarbayNavarro2013} argued that $\mathcal H n - O(n)$ 
	comparisons are necessary if the elements in the input are all distinct.
\end{proof}

\subsection{Powersort: A Single-Pass Stack-Based Method}

One little blemish remains in \proc{PeekSort}: we have to use ``random accesses'' into the
middle of the array to decide how to proceed. 
Even though we only use cache-friendly sequential scans, 
the I/O operations to load the middle run are effectively wasted, 
since it will often be merged only much later (after further recursive calls).
Timsort and the other stack-based variants from~\cite{AugerNicaudPivoteau2015,BussKnop2018}
proceed in one left-to-right scan over the input and merge the top runs 
on their stack. 
This increases the likelihood to still have (parts of) the most recently detected run 
in cache when it is merged subsequently.

\subsubsection{The power of top-down in a bottom-up method}

Method~2 to construct nearly-optimal search trees suggests the following definition:
\begin{definition}[Node Power]
\label{def:node-power}
	Let $\alpha_0,\ldots,\alpha_m$, $\sum \alpha_j=1$ be leaf probabilities. 
	For $1\le j\le m$,
	let \internalnode j be the internal node separating the $(j-1)$st and $j$th leaf.
	The \emph{power} of (the split at) node \internalnode j is
	\begin{align*}\SwapAboveDisplaySkip
			P_j
		&\wrel=
			\min \Biggl\{ 
				\ell\in\N : \biggl\lfloor \frac a{2^{-\ell}} \biggr\rfloor 
					< \biggl\lfloor \frac b{2^{-\ell}} \biggr\rfloor  
			\Biggr\},\;
	\text{where }
			a
		\wrel=
			\sum_{i=0}^{j-1} \alpha_i - \tfrac12 \alpha_{j-1},\;
			b
		\wrel=
			\sum_{i=0}^{j-1} \alpha_i + \tfrac12 \alpha_{j}.
	\end{align*}
	($P_j$ is the index of the first bit where the (binary) fractional parts of $a$ and $b$ differ.)
\end{definition}
Intuitively, $P_j$ is the ``intended'' depth of \internalnode j, but 
nodes occasionally end up higher in 
the tree if some leaf has a large weight relative to the current subtree,
(see the rightmost branch in \wref{fig:nearly-optimal-trees}).
Mehlhorn's~\cite{Mehlhorn1977,Mehlhorn1984} 
original implementation of Method~2, procedure \emph{construct-tree},
does not single out the case that the next desired cut point lies
\emph{outside} the range of a subtree.
This reduces the number of cases, but for our application, it is more convenient 
to explicitly check for this out-of-range case, and if it occurs to
directly proceed to the next cut point.
We refer to the modified algorithm as \emph{Method~2\/${}'$};
\wref{app:method-2-prime} gives the details and shows that the changes do not affect
the guarantee for Method~2 in  \wref{thm:nearly-opt-trees}.
In fact Method~2${}'$ seems to typically yield slightly \emph{better} trees 
than Method~2, but there are also counterexamples.

The core benefit of Method~2${}'$, however, is that the resulting tree is 
characterized by local information, namely the node powers,
without requiring any coordination of a top-down recursive procedure.
\begin{lemma}[Path power monotonicity]
\label{lem:path-monotonicity}
	Consider the tree constructed by Method~2\/${}'$.
	The powers of internal nodes along any root-to-leaf path is strictly increasing.
\end{lemma}
The proof is given in \wref{app:method-2-prime}.
\begin{corollary}
\label{cor:cartesian-tree}
	The tree constructed by Method~2\/${}'$ for leaf probabilities $\alpha_0,\ldots,\alpha_m$
	is the (min-oriented) Cartesian tree for the sequence of node powers $P_1,\ldots,P_m$.
	It can thus be constructed iteratively (left to right) by the algorithm of 
	Gabow, Bentley, and Tarjan~\cite{GabowBentleyTarjan1984}.
\end{corollary}

\subsubsection{Merging on-the-fly}

We implicitly use the observation from \wref{cor:cartesian-tree} in our algorithm ``powersort''
to construct the tree from left to right.
Whenever the next internal node has a smaller power than the preceding one, 
we are following an edge from a left child 
up to its parent. 
That means that this subtree does not change anymore and we can execute any pending merges 
in it before continuing.
If we are instead following an edge down to a right child of a node, that subtree is still ``open''
and we push the corresponding run on the stack.
\wref{alg:powersort} shows the detailed code.

\begin{algorithm}[tbph]
\vspace*{-2ex}
\small
\tikzset{every node/.style={font=\scriptsize}}
	\begin{codebox}
		\Procname{$\proc{PowerSort}(A[1..n])$}
		\li $X \gets $ stack of runs   \>\>\>\>\>\>\textsmaller{(capacity $\lfloor\lg(n)\rfloor+1$)}
		\li $P \gets $ stack of powers \>\>\>\>\>\>\textsmaller{(capacity $\lfloor\lg(n)\rfloor+1$)}
		\li $s_1 \gets 1$; \; 
			$e_1 = \proc{ExtendRunRight}(A[1],n)$
			\quad\Comment $A[s_1..e_1]$ is leftmost run
		\li \While $e_1 < n$
		\Do
			\li $s_2 \gets e_1+1$; \;
				$e_2 \gets \proc{ExtendRunRight}(A[s_2],n)$
				\quad\Comment $A[s_2..e_2]$ next run
			\li $p \gets \proc{NodePower}(s_1,e_1,s_2,e_2,n)$ 
			\quad \Comment $P_j$ for node \internalnode j between $A[s_1..e_1]$ and $A[s_2..e_2]$
			\li \While $P.top() > p$ \>\>\>\>\quad\Comment previous merge deeper in tree than current
			\Do
				\li $P.pop()$
					\>\>\>\quad\Comment $\leadsto$ merge and replace run~$A[s_1..e_1]$ by result
				\li $(s_1,e_1) \gets \proc{Merge}(X.\id{pop}(), \, A[s_1..e_1])$
			\End
			\li $X.\id{push}(A[s_1,e_1])$; \; $P.\id{push}(p)$
			\li $s_1 \gets s_2$; \; $e_1 \gets e_2$
		\EndWhile
		\Comment Now $A[s_1..e_1]$ is the rightmost run
		\li \While $\neg X.\id{empty}()$
		\Do
			\li $(s_1,e_1) \gets \proc{Merge}(X.\id{pop}(), \, A[s_1..e_1])$
		\End
	\end{codebox}
	\vspace*{-5ex}
	\begin{codebox}
		\Procname{$\proc{NodePower}(s_1,e_1,s_2,e_2,n)$}
		\li $n_1 \gets e_1-s_1 + 1$; \; 
			$n_2 \gets e_2 - s_2 + 1$; \; $\ell \gets 0$
		\li $a \gets ( s_1 + n_1/2 - 1 ) / n$; \;
			$b \gets ( s_2 + n_2/2 - 1 ) / n$
		\li \kw{while} $\lfloor a\cdot 2^\ell \rfloor \isequal \lfloor b\cdot 2^\ell \rfloor$ 
			\kw{do} $\ell\gets\ell+1$ \kw{end while}
		\li \Return$(\ell)$
	\end{codebox}
	\vspace*{-1ex}
	\caption{%
		Powersort: A one-pass stack-based nearly-optimal natural mergesort.
		Procedure \proc{ExtendRunRight} scans right as long as the run continues.
	}
	\label{alg:powersort}
	\vspace*{-3ex}
\end{algorithm}

The runs on the stack correspond to nodes with strictly increasing powers,
so we can bound the stack height by the maximal $P_j$.
Since our leaf probabilities here are $\alpha_j = \frac{L_j}n\ge \frac1n$, 
we have $P_j \le \lfloor \lg n\rfloor + 1$.

\begin{theorem}
\label{thm:powersort}
	The merge cost of \proc{PowerSort} is at most $\mathcal H(\frac{L_1}n,\ldots,\frac{L_m}n)+2n$,
	the number of comparisons is at most $\mathcal H(\frac{L_1}n,\ldots,\frac{L_m}n)+3n-r$.
	Moreover, (apart from buffers for merging) only $O(\log n)$ words of extra space 
	are required.
\end{theorem}
\begin{proof}
	The merge tree of \proc{PowerSort} is exactly the search tree constructed by Method~2${}'$
	on leaf probabilities $(\alpha_0,\ldots,\alpha_m) = (\frac{L_1}n,\ldots,\frac{L_r}n)$
	and $\beta_j=0$.
	By \wref{thm:nearly-opt-trees}--(iii), the search costs are 
	$C \le \mathcal H+2$ with $\mathcal H = \mathcal H(\frac{L_1}n,\ldots,\frac{L_m}n)$, so
	the overall merge costs are 
	$M = C n \le \mathcal H n + 2n$, which is within $O(n)$ of the lower bound for $M$.
	The comparisons follow as for \wref{thm:peeksort}.
\end{proof}

\section{Running-Time Study}
\label{sec:experiments}

We conducted a running-time study comparing the
two new nearly-optimal mergesorts with 
current state-of-the-art implementations and
elementary mergesort variants.
The code is available on github~\cite{githubCode}.

\needspace{3\baselineskip}
\smallskip\noindent
The goal of this study is to show that
\begin{enumerate}
\item 
	peeksort and powersort have very little overhead compared to standard (non-natural)
	mergesort variants (\ie, they are never (much) slower), and at the same time
\item
	peeksort and powersort outperform other mergesort variants on
	partially presorted inputs.
\end{enumerate}
Timsort is arguably the most used adaptive sorting method;
even though analytical guarantees are still to be found,
its efficiency in particular for partially sorted inputs has been 
demonstrated empirically~\cite{Peters2002}.
A secondary goal is hence to
\begin{enumerate}
\setcounter{enumi}{2}
\item
	investigate the typical merge costs of Timsort on different inputs.
\end{enumerate}

\subsection{Setup}
Oracle's Java runtime library includes a highly tuned Timsort implementation;
to be able to directly compare with it, we chose to implement our algorithms in Java.
The Timsort implementation is used for \texttt{Object[]}, \ie,
arrays of \emph{references} to objects;
since the location of objects on the heap is hard to control and likely to produce
more or less cache misses from run to run, we chose to sort \texttt{int[]}s instead
to obtain more reproducible results.
We thus modified the library implementation of Timsort accordingly.
This scenario makes key comparisons and element moves relatively cheap and thereby 
emphasizes the remaining overhead of the methods,
which is in line with
our primary goal 
	1) to study the impact of the additional bookkeeping
	required by the adaptive methods.
We compare our methods with simple top-down and bottom-up mergesort implementations.
We use the code given by Sedgewick~\cite[Programs~8.3 and~8.5]{Sedgewick2003}
with his simple merge method (Program~8.2) as its basis;
in both cases, we add a check before calling merge to detect if the runs happen 
to be already in sorted order, and we use insertionsort for base cases of size $n\le w = 24$.
(For bottom-up mergesort, we start by sorting chunks of size $w=24$.)
Our Java implementations of peeksort and powersort are described in more detail in \wref{app:java-code}.
Apart from a mildly optimized version of the pseudocode, we added the same cutoff / minimal run length 
($w=24$) as above.

All our methods call the same merge procedure, whereas
the library Timsort contains a modified merge method
that tries to save key comparisons:
when only elements from the same run are selected for the output repeatedly,
Timsort enters a \emph{``galloping mode''} and uses exponential searches 
(instead of the conventional sequential search) to find
the insertion position of the next element. Details are described by Peters~\cite{Peters2002}.
Since saving comparisons is not of utmost importance in our scenario of sorting \texttt{int}s,
we also added a version of Timsort, called ``trotsort'', that uses our plain merge method instead
of galloping, but is otherwise identical to the library Timsort.

Our hard- and software setup is listed in
\wref{app:setup}.
We use the following inputs types:
\begin{itemize}
\item
	\emph{random permutations} are a case where no long runs are present to exploit;
\item
	\emph{``random-runs'' inputs} are constructed from a random permutation by sorting 
	segments of random lengths,
	where the lengths are chosen independently according to a geometric distribution with a 
	given mean $\ell$;
	since the geometric distribution has large variance, 
	these inputs tend to have runs whose sizes vary a lot;
\item
	\emph{``Timsort-drag'' inputs} are special instances of random-runs inputs 
	where the run lengths are chosen as $\mathcal R_{\mathrm{tim}}$,
	the bad-case example for Timsort from~\cite[Thm.\,3]{BussKnop2018}.
\end{itemize}

\subsection{Overhead of Nearly-Optimal Merge Order}

We first consider random permutations as inputs.
Since random permutations contain (with high probability) no long runs 
that can be exploited, the adaptive methods will not find anything 
that would compensate for their additional efforts to identify runs.
(This is confirmed by the fact that the total merge costs of all methods, 
including Timsort, are within 1.5\% of each other in this scenario.)
\wref{fig:times-normalized-rp} shows average running times
for inputs sizes from $100\,000$ to $100$ million ints.
(Measurements for $n=10\,000$ were too noisy to draw meaningful conclusions.)

\begin{figure}[tbh]
	\begin{minipage}[c]{0.6\textwidth}
		\hspace*{-.5em}\includegraphics[width=\textwidth]{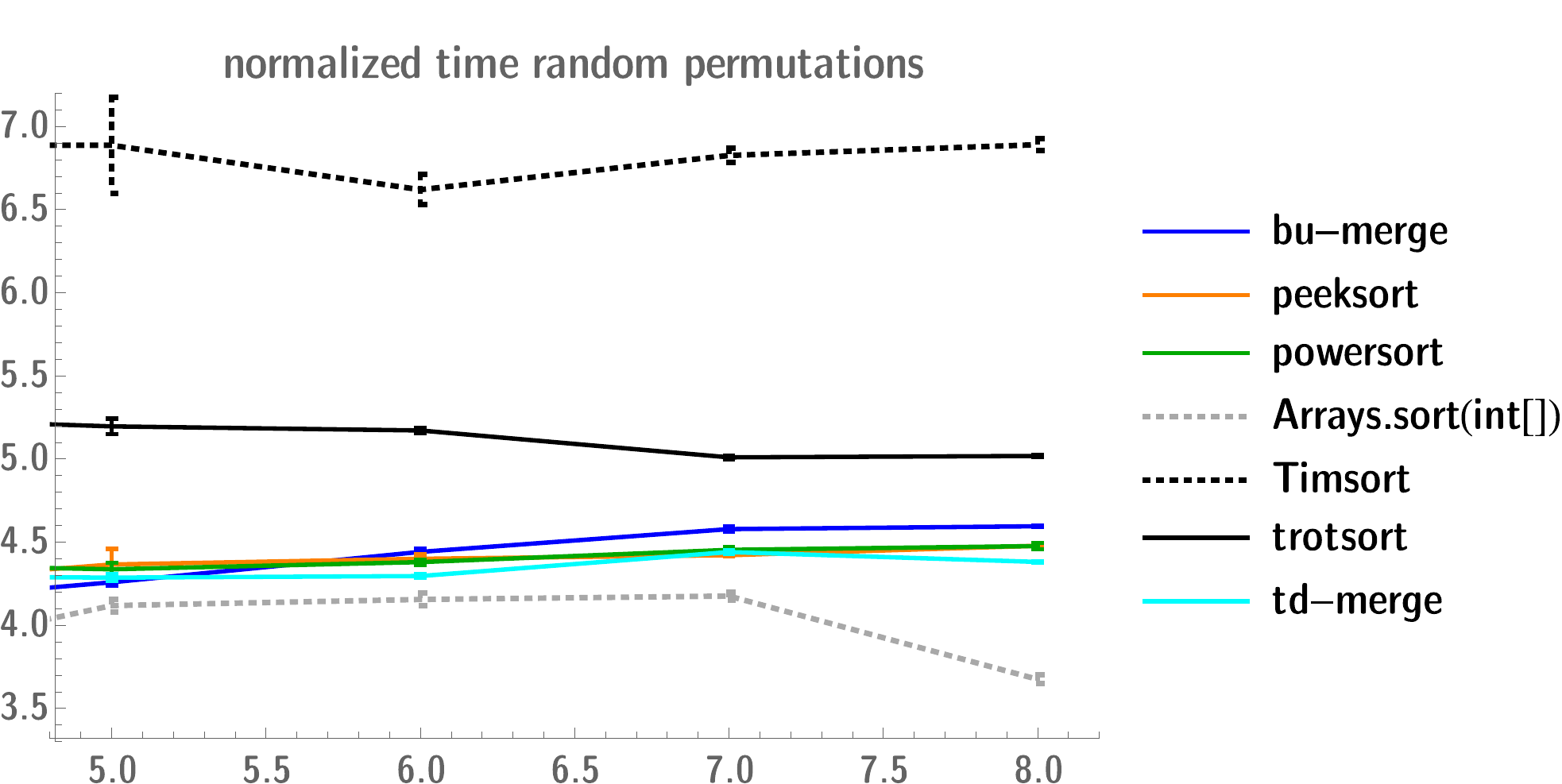}%
  	\end{minipage}\hfill
  	\begin{minipage}[c]{.38\textwidth}
		\caption{%
			Normalized running times on random permutations. 
			The $x$-axis shows $\log_{10}(n)$, the $y$-axis show average and standard deviation (as error bars)
			of $t/(n \lg n)$ where $t$ is the running time in ms.
			We plot averages over 1000 repetitions (resp.\ 200 and 20 for the largest sizes).
		}
	\label{fig:times-normalized-rp}
	\end{minipage}
\end{figure}

The relative ranking is clearly stable across different input sizes.
\texttt{Arrays.sort(int[])} (dual-pivot quicksort) is the fastest method,
but is not a stable sort and only serves as a baseline.
The timings of top-down and bottom-up mergesort, 
peeksort and powersort are 20--30\% slower than dual-pivot quicksort.
Comparing the four to each other, no large differences
are visible; if anything, bottom-up mergesort was a bit slower for large inputs.
Since the sorting cutoff / minimal run length $w=24$ exceeded the length of \emph{all} 
runs in all inputs, we are effectively presented with a case of all-equal run lengths.
Merging them blindly from left to right (as in bottom-up mergesort) then performs just fine,
and top-down mergesort finds a close-to-optimal merging order in this case.
That peek- and powersort perform essentially \emph{as good as} elementary mergesorts
on random permutations thus clearly indicates that their overhead 
for determining a nearly-optimal merging order is negligible.

The library Timsort performs surprisingly bad on \texttt{int[]}s,
probably due to the relatively cheap comparisons.
Replacing the galloping merge with the ordinary merge alleviates this
(see ``trotsort''), 
but Timsort remains inferior on random permutations by a fair margin
(10--20\%).

\subsection{Practical speedups by adaptivity}

	\begin{figure}[tbh]
		\plaincenter{
			\includegraphics[width=.8\linewidth]{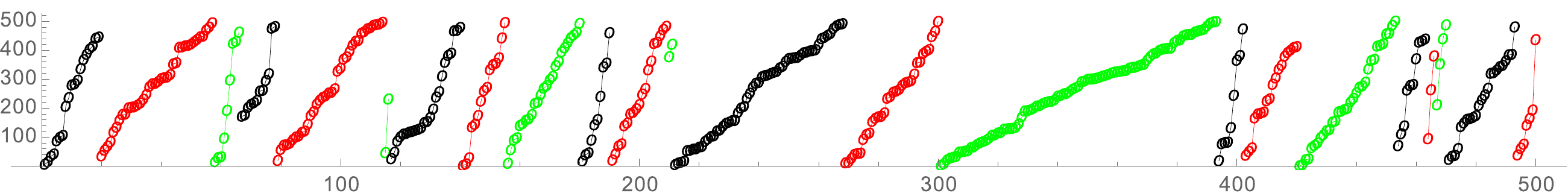}
		}
		\caption{%
			A random-runs input with $n=500$ and $\ell=20\approx\sqrt n$. The $x$-axis is the index in the array,
			the $y$-axis the value. Runs are emphasized by color.%
		}
		\label{fig:typical-random-runs}
	\end{figure}
After demonstrating that we do not lose much by using our adaptive methods
when there is nothing to adapt to,
we next investigate how much can be gained if there is.
We consider random-runs inputs as described above.
This input model instills a good dose of presortedness, but
not in a way that gives any of the tested methods an obvious
advantage or disadvantage over the others.
We choose a representative size of $n=10^7$ and an expected run length $\ell = 3\,000$,
so that we expect roughly $\sqrt n$ runs of length $\sqrt n$.

\begin{figure}[tbh]
	\plaincenter{%
			\includegraphics[height=10\baselineskip]{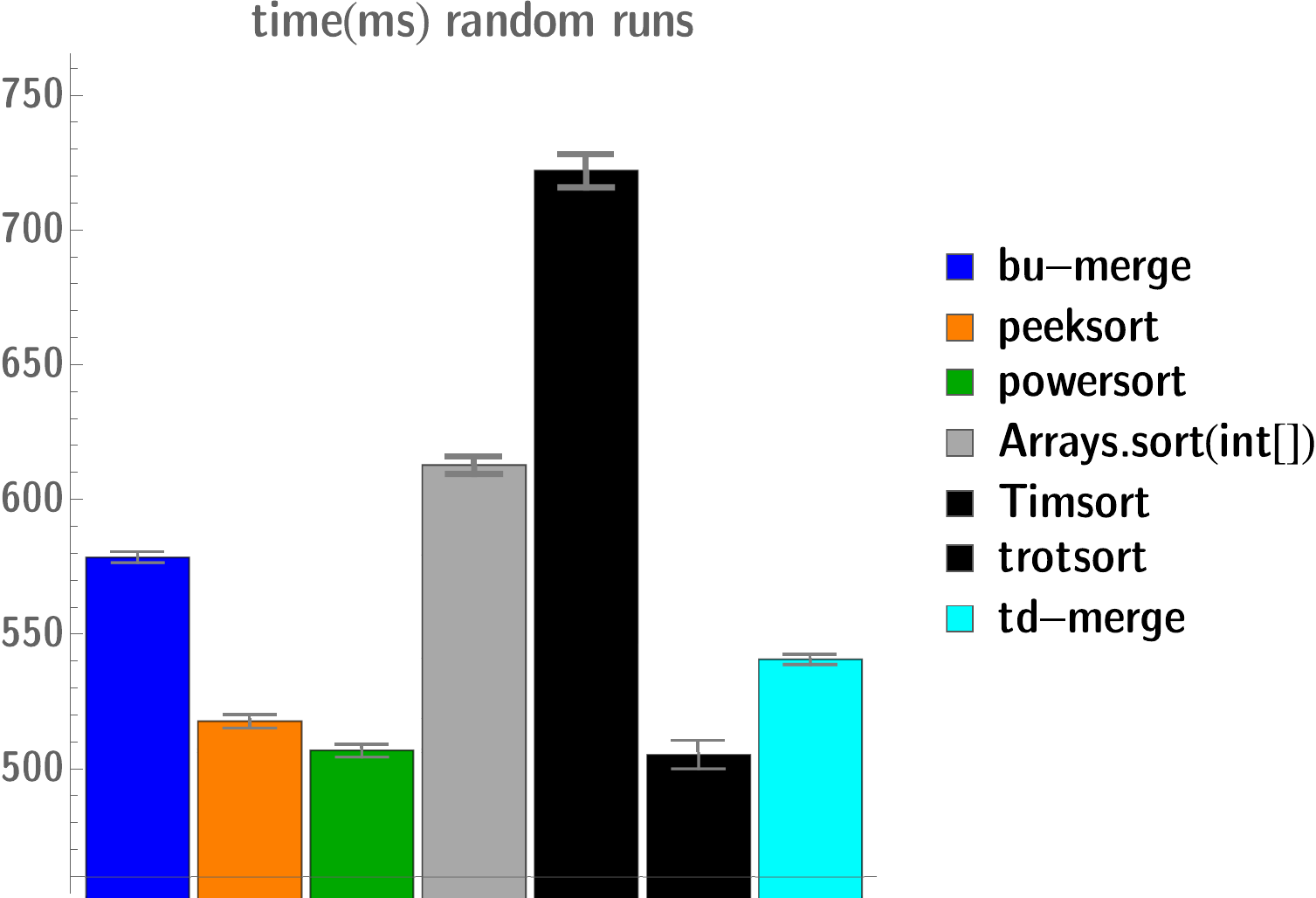}\qquad%
			\includegraphics[height=10\baselineskip]{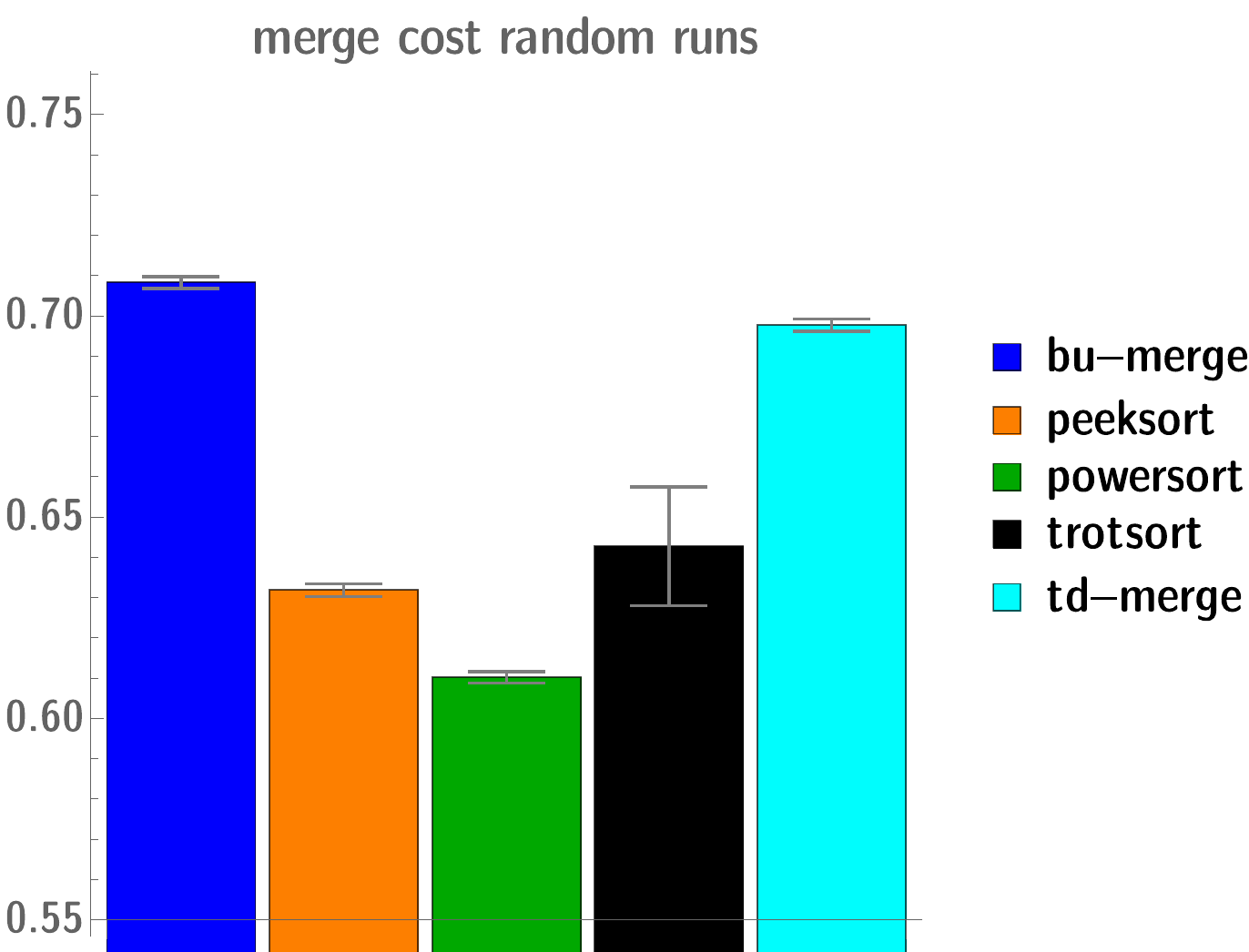}%
	}
	\vspace{-1ex}
	\caption{%
		Average running times (left) and normalized merge cost (right) on 
		random-runs inputs with $n=10^7$ and $\ell=3\,000\approx \sqrt n$.
		The merge costs have been divided by $n\lg(n/w)\approx1.87\cdot 10^8$, 
		which is the merge cost a (hypothetical) optimal mergesort that does 
		not pick up existing runs, but starts merging with runs of length $w=24$.
	}
\label{fig:times-mergecost-random-runs}
\end{figure}

If this was a random permutation, we would expect merge costs of 
roughly $n \lg(n/w) \approx 1.87 \cdot10^8$ (indeed a bit above this).
The right chart in \wref{fig:times-mergecost-random-runs} shows 
that the adaptive methods can bring the merge cost down to a little over 60\%
of this number.
	Note the run lengths vary considerably~-- 
	to give some intuitive feel for this volatility, \wref{fig:typical-random-runs} 
	shows a stereotypical (but smaller) random-runs input.
Powersort achieved average merge costs of 
$1.14\cdot10^8 < n \lg r \approx 1.17 \cdot 10^8$, \ie,
less than a method would that
only adapts to the \emph{number} of runs $r$.

In terms of running time, powersort is again among the fastest stable methods,
and indeed 20\% \emph{faster} than \texttt{Arrays.sort(int[])}.
The best adaptive methods are also 10\% faster than top-down mergesort.
(The latter is ``slightly adaptive'', 
by omitting merges if the runs happen to already be in order.)
This supports the statement that significant speedups can be realized
by adaptive sorting on inputs with existing order, and 
$\sqrt n$ runs suffice for that.
If we increase $\ell$ to $100\,000$, so that we expect only roughly $100$ long runs,
the library quicksort becomes twice as slow as powersort and Timsort (trotsort).

Timsort's running time is a bit anomalous again. 
Even though it occasionally incurs 10\% more merge costs on a given input than
powersort, the running times were within 1\% of each other 
(considering the trotsort variant;
the original galloping version was again uncompetitive).

\subsection{Non-optimality of Timsort}

Finally, we consider ``Timsort-drag'' inputs, a sequence
of run lengths $\mathcal R_{\mathrm{tim}}(n)$ 
specifically crafted by Buss and Knop~\cite{BussKnop2018}
to generate unbalanced merges (and hence large merge cost) in Timsort.
Since actual Timsort implementations employ minimal run lengths of up to $32$ elements
we multiplied the run lengths by $32$.
\wref{fig:times-mergecost-timdrag} shows running time and merge cost for all
methods on a characteristic Timsort-drag input of length $2^{24}\approx 1.6 \cdot 10^7$.

\begin{figure}[tbh]
	\plaincenter{%
			\includegraphics[height=10\baselineskip]{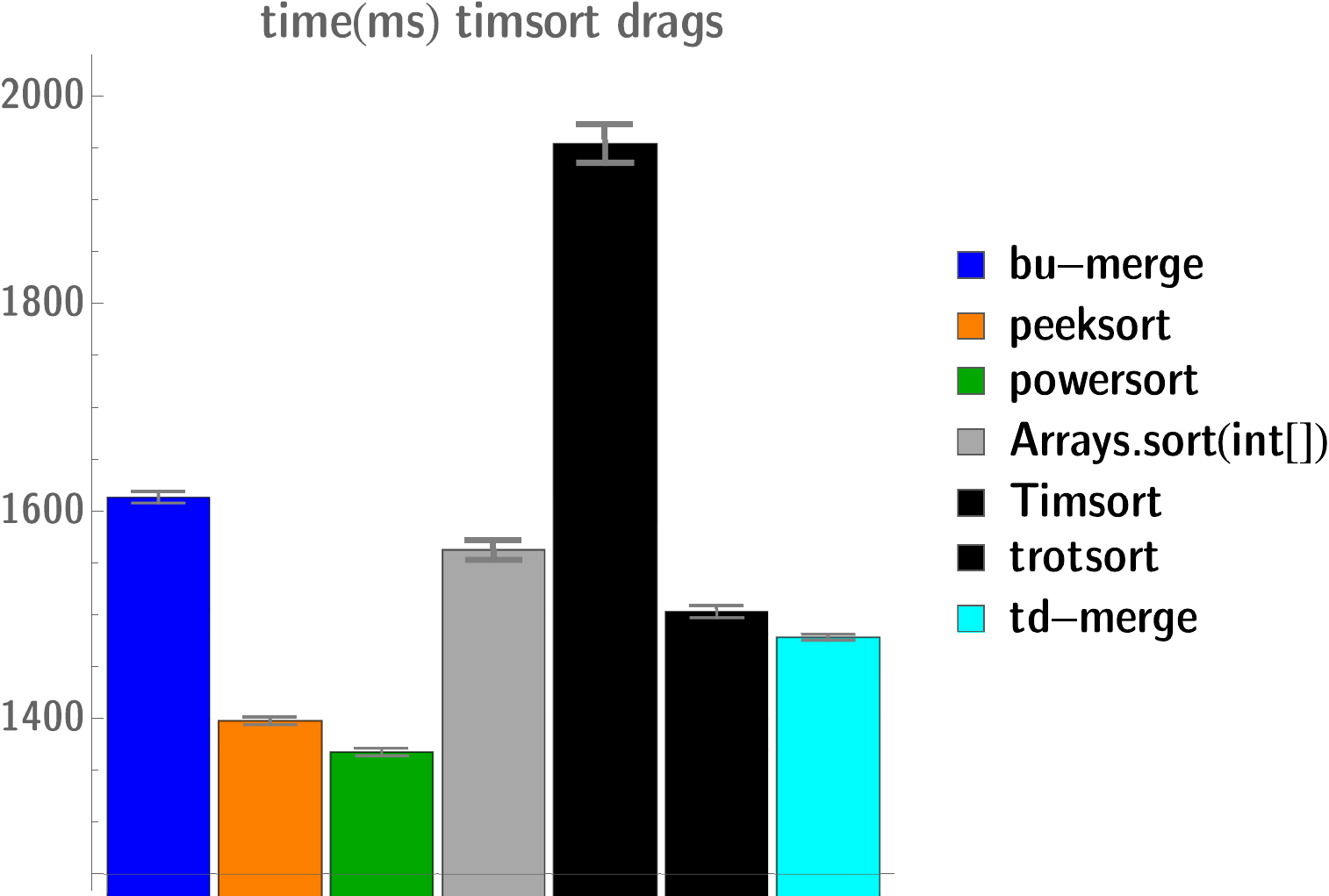}\qquad%
			\includegraphics[height=10\baselineskip]{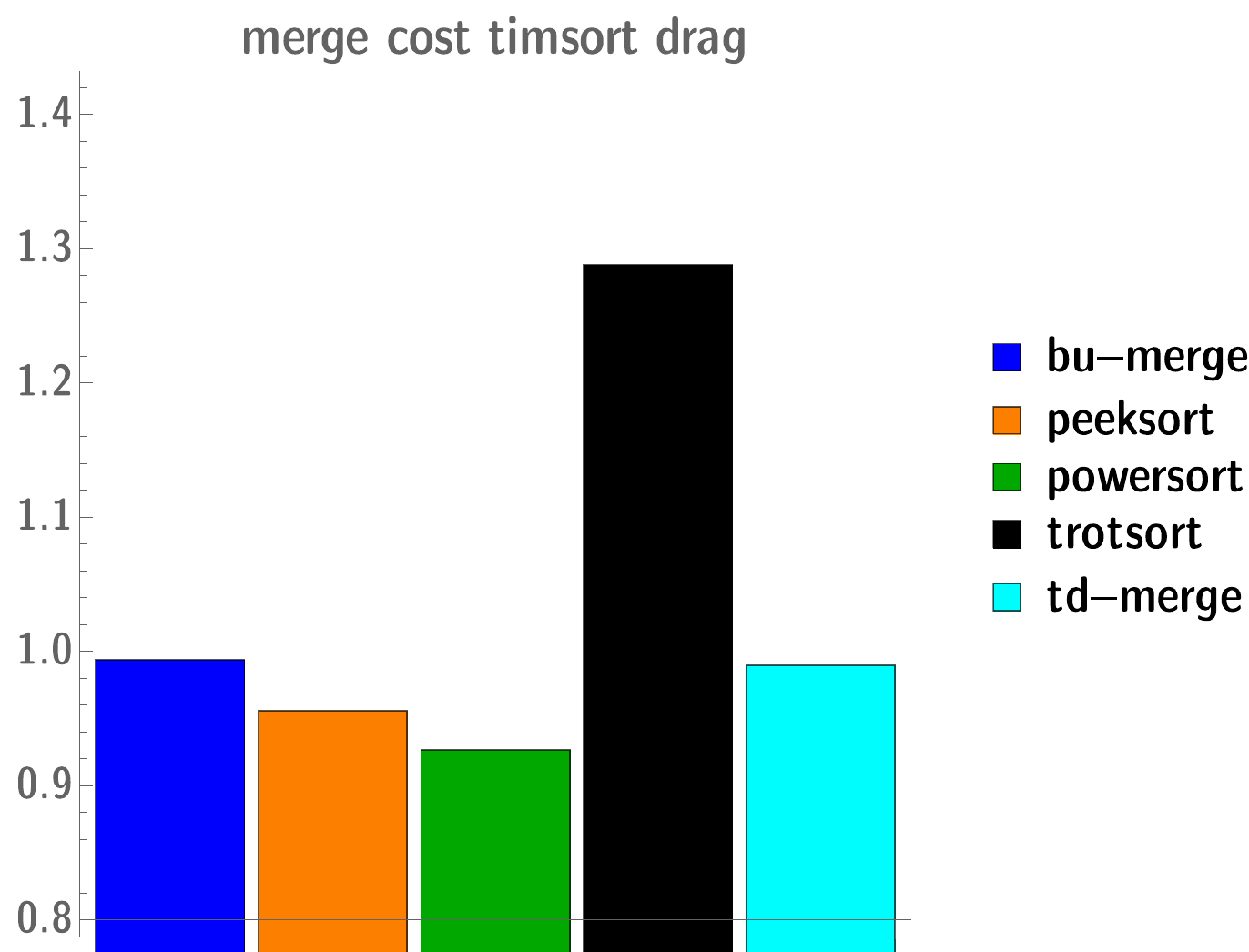}%
	}
	\vspace{-1ex}
	\caption{%
		Average running times (left) and normalized merge cost (right) on 
		``Timsort-drag'' inputs with $n=2^{24}$ and run lengths 
		$\mathcal R_{\mathrm{tim}}(2^{24}/32)$ multiplied by $32$.
		Merge costs have been divided by $n \lg(n/w) \approx 3.26\cdot10^8$.
	}
\label{fig:times-mergecost-timdrag}
\end{figure}

In terms of merge costs, Timsort now pays 30\% more than even a simple non-adaptive
mergesort, whereas peeksort and powersort obviously retain their proven nearly-optimal
behavior.
Also in terms of running time, Timsort is a bit slower than top-down mergesort,
and 10\% slower than powersort on these inputs.
It is remarkable that the dramatically larger merge cost does not lead to 
a similarly drastic slow down in practice.
Nevertheless, it must be noted that Timsort's merging strategy has weaknesses,
and it is unclear if more dramatic examples are yet to be found.

\section{Conclusion}
\label{sec:conclusion}

In this paper, we have demonstrated that provably good merging orders for natural mergesort 
can be found with negligible overhead.
The proposed algorithms, peeksort and powersort offer more reliable performance than
the widely used Timsort, and at the same time, are arguably simpler.

Powersort builds on a modified bisection heuristic for computing nearly-optimal binary search trees
that might be independent interest.
It has the same quality guarantees as Mehlhorn's original formulation, but
allows the tree to be built ``bottom-up'' as a Cartesian tree over a certain sequence, the ``node powers''.
It is the only such method for nearly-optimal search trees to our knowledge.

Buss and Knop conclude with the question
whether there exists a $k$-aware algorithm 
(a stack-based natural mergesort that only considers the top $k$ runs in the stack)
with merge cost $(1+o_r(1))$ times the optimal merge cost~\cite[Question~37]{BussKnop2018}.
Powersort effectively answers this question in the affirmative with $k=3$.%
\footnote{%
	Strictly speaking, powersort needs a relaxation of the model of Buss and Knop.
	They require decisions to be made solely based on the \emph{lengths} of runs, 
	whereas node power takes the location of the runs within the array into account.
	Since the location of a run must be stored anyway, this appears reasonable to us.
}

\subsection{Extensions and future work}

Timsort's ``galloping merge'' procedure saves comparisons when we consistently consume 
elements from one run, but in ``well-mixed'' merges, it does not help (much).
It would be interesting to compare this method with other comparison-efficient
merge methods.

Another line of future research is to explore ways to profit from duplicate keys in the input.
The ultimate goal would be a ``synergistic'' sorting method
(in the terminology of~\cite{BarbayOchoaSatti2017}) that has practically no overhead for detecting
existing runs and equals and yet exploits their \emph{combined} presence optimally.

\clearpage
\appendix

\section{Method~2${}'$ for nearly-optimal search trees}
\label{app:method-2-prime}

This section uses Mehlhorn's~\cite{Mehlhorn1977,Mehlhorn1984} original notation,
in particular, it handles the general optimal binary search tree setting.
We add the following two ``out-of-range'' cases to Mehlhorn's procedure 
$\mathit{construct\text-tree}(i,j,\mathit{cut},\ell)$, at the very beginning:

\begin{itemize}
\item (Case L) If $s_j < \mathit{cut} + 2^{-\ell}$, return
	$\mathit{construct\text-tree}(i,j,\mathit{cut},\ell+1)$
\item (Case R) If $\mathit{cut} + 2^{-\ell} < s_i$, return
	$\mathit{construct\text-tree}(i,j,\mathit{cut} + 2^{-\ell},\ell+1)$
\end{itemize}

The original procedure would end up in case C instead of L resp.\ 
case D instead of R and choose a right- resp.\ leftmost key as the
root node.
But if the desired cut point $\mathit{cut}+2^{-\ell}$ lies completely
outside the range of bisection $[s_i,s_j]$, this produces an unnecessarily
unbalanced split.
This case can only happen if the neighboring leaf has a probability
larger than $\frac12$ relative to the previous subtree, so that the current split point
still lies within the range corresponding to the previously chosen root.
Our cases L and R thus ``skip'' this void cut point $\mathit{cut}+2^{-\ell}$
and increment $\ell$ without generating a node.

Given that the invariants (1)--(4) are fulfilled for the current parameters
of \emph{construct-tree}, they will also be fulfilled in the recursive calls
issued in cases L and R.
Therefore, Mehlhorn's analysis remains valid for our modified procedure:
In his Fact~3, we have $b_h+1 \le \ell$ and $a_h \le \ell$ (instead of equality),
but this is all that is needed to establish Fact~4 and hence the bound on the search costs.

\subsection{Proof of \wref{lem:path-monotonicity}}

\textbf{\textsf{Claim:}} 	
In the tree constructed by Method~2\/${}'$,
the powers of internal nodes along any root-to-leaf path are strictly increasing.

\begin{proof}
Consider the recursive procedure \emph{construct-tree}
as described by Mehlhorn~\cite{Mehlhorn1977,Mehlhorn1984}, 
but with the additional cases from above.
We prove that whenever a recursive call 
$\textit{construct\text-tree}(i,j,\mathit{cut},\ell)$ creates
an internal node \internalnode k, we have $P_k = \ell$.
Since $\ell$ is incremented in all recursive calls, the claim follows.

Only cases A, B, C or D create new nodes, so assume we are not in case L or R.
Then we actually have the stronger version of invariant~(4):
$\mathit{cut} \le s_i\le \mathit{cut} + 2^{-\ell}\le s_j \le \mathit{cut}+2^{-\ell+1}$
and hence will always find a $k$ with $i<k\le j$ and 
\[\mathit{cut} \wrel\le s_{k-1} \wrel< \mathit{cut}+2^{-\ell} \wrel\le s_k \wrel\le \mathit{cut}+2^{-(\ell-1)}\]
for which we now create the node \internalnode k.
Dividing by $2^{-\ell}$ shows that $P_k = \ell$ (we have $a = s_{k-1}$ and $b = s_k$).
\end{proof}

\section{Java Code}
\label{app:java-code}
\lstset{
	tabsize=3,
	basewidth=.55em,
	language=Java,
	numbers=left,
	numberstyle=\tiny\sffamily,
	numbersep=-.8em,
	basicstyle=\footnotesize\ttfamily,
}

\bigskip\noindent
In this appendix, we give the core methods of a Java implementation of
\wref{alg:peeksort} and \wref{alg:powersort} that form the basis of our running-time
study.
The goal is to document (and comment on) a few design decisions and optimizations
that have been included in the running time study.
The full code, including the main routines to reproduce the running time studies
exactly as used for this paper,
is available on github~\cite{githubCode}.

\subsection{Merging}

All our mergesort variants (except for the original library implementation of Timsort)
use the following straight-forward ``bitonic'' merge procedure that
merges two adjacent runs.
It is taken from \cite[Program~8.2]{Sedgewick2003}

\begin{lstlisting}
	/** merges A[l..m-1] and A[m..r] */
	static void mergeRuns(int[] A, int l, int m, int r, int[] aux) {
		--m; // accounts for different convention in Sedgewick's book
		int i, j;
		for (i = m+1; i > l; --i) aux[i-1] = A[i-1];
		for (j = m; j < r; ++j) aux[r+m-j] = A[j+1];
		for (int k = l; k <= r; ++k)
			A[k] = aux[j] < aux[i] ? aux[j--] : aux[i++];
	}
\end{lstlisting}

Merging offers some potential for improvements in particular \wrt the number
of used key comparisons.
Since we operate on integers, comparisons are cheap and more sophisticated
merging strategies will not be needed here.

\subsection{Peeksort}

The pseudocode for Peeksort can be translated to Java almost verbatim.
Since the recursion is costly for small input sizes,
we switch to insertionsort when the subproblem size
is small. 
Values for \texttt{INSERTION\_SORT\_THRESHOLD} of $10$--$32$ yielded good 
results in the experiments; we ultimately set it to $24$ to approximate
the choice in the library implementation of Timsort.
(The latter chooses the minimal run length between $16$ and $32$ 
depending on $n$ by a heuristic that tries to avoid imbalanced merges.
We did not use this adaptive choice in our methods.)

\begin{lstlisting}[language=Java]
	static void peeksort(int[] A, int left, int right, 
	                     int leftRunEnd, int rightRunStart, int[] B) {
		if (leftRunEnd == right || rightRunStart == left) return;
		if (right - left + 1 <= INSERTION_SORT_THRESHOLD) {
			insertionsort(A, left, right, leftRunEnd - left + 1); return;
		}
		int mid = left + ((right - left) >> 1);
		if (mid <= leftRunEnd) {            // |XXXXXXXX|XX     X|
			peeksort(A, leftRunEnd+1, right, leftRunEnd+1,rightRunStart, B);
			mergeRuns(A, left, leftRunEnd+1, right, B);
		} else if (mid >= rightRunStart) {  // |XX     X|XXXXXXXX|
			peeksort(A, left, rightRunStart-1, leftRunEnd, rightRunStart-1, B);
			mergeRuns(A, left, rightRunStart, right, B);
		} else { // find middle run
			int i, j;
			if (A[mid] <= A[mid+1]) {
				i = extendIncreasingRunLeft(A, mid, leftRunEnd + 1);
				j = mid+1 == rightRunStart ? mid : 
					extendIncreasingRunRight(A, mid+1, rightRunStart - 1);
			} else {
				i = extendDecreasingRunLeft(A, mid, leftRunEnd + 1);
				j = mid+1 == rightRunStart ? mid : 
					extendStrictlyDecreasingRunRight(A, mid+1,rightRunStart - 1);
				reverseRange(A, i, j);
			}
			if (i == left && j == right) return;
			if (mid - i < j - mid) {         // |XX     x|xxxx   X|
				peeksort(A, left, i-1, leftRunEnd, i-1, B);
				peeksort(A, i, right, j, rightRunStart, B);
				mergeRuns(A,left, i, right, B);
			} else {                         // |XX   xxx|x      X|
				peeksort(A, left, j, leftRunEnd, i, B);
				peeksort(A, j+1, right, j+1, rightRunStart, B);
				mergeRuns(A,left, j+1, right, B);
			}
		}
	}
\end{lstlisting}

\subsection{Insertionsort to extend runs}

Here (and in the following code), we use a straight insertionsort
variant that accepts the length of a sorted prefix as an additional parameter.
A similar method is also used in Timsort to extend runs to a forced minimal
length.
The library Timsort uses binary insertionsort instead, 
but unless comparisons are expensive, a straight sequential-search variant is sufficient.

\begin{lstlisting}
	static void insertionsort(int[] A, int left, int right, int nPresorted) {
		assert right >= left;
		assert right - left + 1 >= nPresorted;
		for (int i = left + nPresorted; i <= right ; ++i) {
			int j = i - 1,  v = A[i];
			while (v < A[j]) {
				A[j+1] = A[j];
				--j;
				if (j < left) break;
			}
			A[j+1] = v;
		}
	}
\end{lstlisting}

\subsection{Powersort}

For powersort, we implement the stack as an array that is indexed by the node power.
Thereby, we avoid explicit stack operations and to store powers explicitly.
On the other hand, we have to check for empty entries since powers are not always consecutive.

\begin{lstlisting}[language=Java]
	static void powersort(int[] A, int left, int right) {
		int n = right - left + 1;
		int lgnPlus2 = log2(n) + 2;
		int[] leftRunStart = new int[lgnPlus2], leftRunEnd = new int[lgnPlus2];
		Arrays.fill(leftRunStart, -1);
		int top = 0;
		int[] buffer = new int[n >> 1];
	
		int startA = left, endA = extendRunRight(A, startA, right);
		int lenA = endA - startA + 1;
		if (lenA < minRunLen) { // extend to minRunLen
			endA = Math.min(right, startA + minRunLen-1);
			insertionsort(A, startA, endA, lenA);
		}
		while (endA < right) {
			int startB = endA + 1, endB = extendRunRight(A, startB, right);
			int lenB = endB - startB + 1;
			if (lenB < minRunLen) { // extend to minRunLen
				endB = Math.min(right, startB + minRunLen-1);
				insertionsort(A, startB, endB, lenB);
			}
			int k = nodePower(left, right, startA, startB, endB);
			assert k != top;
			for (int l = top; l > k; --l) { // clear left subtree bottom-up
				if (leftRunStart[l] == -1) continue;
				mergeRuns(A, leftRunStart[l], leftRunEnd[l]+1, endA, buffer);
				startA = leftRunStart[l];
				leftRunStart[l] = -1;
			}
			// store left half of merge between A and B on stack
			leftRunStart[k] = startA; leftRunEnd[k] = endA;
			top = k;
			startA = startB; endA = endB;
		}
		assert endA == right;
		for (int l = top; l > 0; --l) {
			if (leftRunStart[l] == -1) continue;
			mergeRuns(A, leftRunStart[l], leftRunEnd[l]+1, right, buffer);
		}
	}
\end{lstlisting}

The computation of the node powers can be done in many different ways and offers a lot
of potential for low-level bitwise optimizations; some care is needed to prevent overflows.
In our experiments, the following loop-less version was a tiny bit faster than
other tried alternatives.

\begin{lstlisting}
	static int nodePower(int left, int right, int startA, int startB, int endB) {
		int twoN = (right - left + 1) << 1; // 2*n
		long l = startA + startB - (left << 1);
		long r = startB + endB + 1 - (left << 1);
		int a = (int) ((l << 31) / twoN);
		int b = (int) ((r << 31) / twoN);
		return Integer.numberOfLeadingZeros(a ^ b);
	}
\end{lstlisting}

\section{Experimental Setup}
\label{app:setup}

All experiments were run on a Lenovo Thinkpad X230 Tablet
running Ubuntu 16.04.01 with Linux kernel 4.13.0-38-generic.
The CPU is an Intel Core i7-3520M CPU with 2.90GHz, the system has 8GB of main memory.

The Java compiler was from the Oracle Java JDK version 1.8.0\_161,
the JVM is Java HotSpot 64-Bit Server VM (build 25.161-b12, mixed mode).
All experiments were run with disabled X server from the TTY,
the java process was bound to one core (with the \texttt{taskset} utility).
They started with a warmup phase to trigger just-in-time (JIT) compilation
before measuring individual sorting operations.
The inputs were generated outside the timing window and reused the same array for all
repetitions.
The following flags were used for the JVM:
\texttt{-XX:+UnlockDiagnosticVMOptions}, \texttt{-XX:-TieredCompilation} and  \texttt{-XX:+PrintCompilation}.
Tiered compilation was disabled to avoid multiple passes of just-in-time compilation
to occur during the timed experiments; the print-compilation flag was used to
monitor whether relevant methods are subjected to recompilation or deoptimization
during the experiments.

\bibliography{mergesort.bib}

\end{document}